\newif\ifsubmit

\submittrue

\ifsubmit                       %
\documentclass[11pt]{article}   %
\else                           %
\documentclass{article}
\fi

\usepackage{article} %
\usepackage{math} %
\usepackage{algo} %
\usepackage{wrapfig} %
\usepackage{placeins} %
\usepackage{graphicx} %
\usepackage[longnamesfirst,numbers]{natbib} %
\usepackage{setspace}

\usepackage[font={small},labelfont=bf,textfont=it]{caption}

\ifsubmit \usepackage[letterpaper]{geometry}
\fi

\allowdisplaybreaks

 \everymath{\displaystyle}

\providecommand{\cost}{c\optsub} %
\providecommand{\capacity}{\cost} %
\providecommand{\sumcost}{\overline{c}\optpar} %
\providecommand{\trees}{\mathcal{T}} %
\newcommand{\forests}{\mathcal{F}} %
\newcommand{\weight}{w\optsub}
\newcommand{\weights}{\mathcal{W}\optsub} %
\newcommand{\leaves}{\mathcal{L}\optsub} %
\newcommand{\apxweight}{\tilde{w}\optsub} %

\newcommand{\subforests}[1][T]{\mathcal{C}_{#1}}

\newcommand{\authoryearp}[1]{[\citeauthor{#1} \citeyear{#1}]}
\newcommand{\authoryear}[1]{\citeauthor{#1} [\citeyear{#1}]}

\begin{document}

\title{Fast and Deterministic Approximations for
  $k$-Cut}

\author{Kent Quanrud\thanks{Department of Computer Science, University
    of Illinois, Urbana-Champaign, Urbana, IL 61801. {\tt
      quanrud2@illinois.edu}.  Work on this paper is partly supported
    by NSF grant CCF-1526799.}}

\maketitle

\begin{abstract}
  In an undirected graph, a $k$-cut is a set of edges whose removal
  breaks the graph into at least $k$ connected components. The minimum
  weight $k$-cut can be computed in $n^{\bigO{k}}$ time, but when $k$
  is treated as part of the input, computing the minimum weight
  $k$-cut is NP-Hard \authoryearp{gh-94}. For $\poly{m,n,k}$-time
  algorithms, the best possible approximation factor is essentially 2
  under the small set expansion hypothesis
  \authoryearp{manurangsi-17}. \authoryear{sv-95} showed that a
  $\parof{2 - \frac{2}{k}}$-approximately minimum weight $k$-cut can
  be computed via $\bigO{k}$ minimum cuts, which implies a
  $\apxO{k m}$ randomized running time via the nearly linear time
  randomized min-cut algorithm of
  \authoryear{karger-00}. \authoryear{nk} showed that the minimum
  weight $k$-cut can be computed deterministically in
  $\bigO{mn + n^2 \log n}$ time. These results prompt two basic
  questions. The first concerns the role of randomization. Is there a
  deterministic algorithm for $2$-approximate $k$-cuts matching the
  randomized running time of $\apxO{k m}$?  The second question
  qualitatively compares minimum cut to 2-approximate minimum
  $k$-cut. Can $2$-approximate $k$-cuts be computed as fast as the
  minimum cut -- in $\apxO{m}$ randomized time?

  We make progress on these questions with a deterministic
  approximation algorithm that computes $\parof{2 + \eps}$-minimum
  $k$-cuts in $\bigO{m \log^3 n / \eps^2}$ time, via a
  $\epsmore$-approximation for an LP relaxation of $k$-cut.
\end{abstract}

\section{Introduction}

Let $G = (V,E)$ be an undirected graph with $m$ edges and $n$
vertices, with positive edge capacities given by $c : E \to
\preals$. A \emph{cut} is a set of edges $C \subseteq E$ whose removal
leaves $G$ disconnected. For $k \in \naturalnumbers$, a \emph{$k$-cut}
is a set of edges $C \subseteq E$ whose removal leaves $G$
disconnected into at least $k$ components. The capacity of a cut $C$
is the sum capacity $\sumcost{C} = \sum_{e \in C} \cost{e}$ of edges
in the cut. The \emph{minimum $k$-cut} problem is to find a $k$-cut
$C$ of minimum capacity $\sumcost{C}$.

The special case $k = 2$, which is to find the minimum cut, is
particularly well-studied. The minimum cut can be computed in
polynomial time by fixing a source $s$ and computing the minimum
$s$-$t$ cut (via $s$-$t$ max-flow) for all choices of
$t$. \citet{ni-90,ni-92-algorithmica,ni-92-sidm} and \citet{ho-94}
improved the running time to $\apxO{m n}$ which, at the time, was as
fast as computing a single maximum flow. A randomized edge contraction
algorithm by \citet{ks-96} finds the minimum cut with high probability
in $\apxO{n^2}$ time; this algorithm is now a staple of graduate level
courses on randomized algorithms. \citet{karger-00} gave a randomized
algorithm based on the Tutte--Nash-Williams theorem \citep{t-61,n-61}
that computes the minimum weight cut with high probability in
$\apxO{m}$ time. The best deterministic running time for minimum cut
is currently $\bigO{m n + n^2 \log n}$, by \citet{sw-97}.  Computing
the minimum capacity cut deterministically in nearly linear time is a
major open problem.  Recently, \citet{kt-15} made substantial progress
on this problem with a deterministic nearly linear time algorithm for
computing the minimum \emph{cardinality} cut in an unweighted simple
graph. This algorithm was simplified by \citet{lst}, and a faster
algorithm was obtained by \citet{hrw-17}.

The general case $k > 2$ is more peculiar. \citet{gh-94} showed that
for any fixed $k$, finding the minimum $k$-cut is polynomial time
solvable, but when $k$ is part of the input, the problem is
NP-Hard. The aforementioned randomized contraction algorithm of
\citet{ks-96} computes a minimum $k$-cut with high probability in
\begin{math}
  \apxO{n^{2(k-1)}}
\end{math}
time. \citet{thorup-08} gave a deterministic algorithm that also
leverages the Tutte--Nash-Williams theorem \citep{t-61,n-61} and runs
in $\apxO{m n^{2k -2}}$ time; this approach was recently refined to
improve the running time to $\apxO{m n^{2k-3}}$ \citep{cqx}. There are
slightly faster algorithms for particularly small values of $k$
\citep{levine} and when the graph is unweighted \citep{gll-18-focs}.
As far as algorithms with running times that are polynomial in $k$ are
concerned, \citet{sv-95} showed that a
$\parof{2 - \frac{2}{k}}$-approximate $k$-cut can be obtained by
$\bigO{k}$ minimum cut computations. By the aforementioned min-cut
algorithms, this approach can be implemented in $\apxO{k m}$
randomized time, $\apxO{k m n}$ time deterministically, and
$\apxO{k m}$ time deterministically in unweighted graphs. Alternatively,
\citet{sv-95} showed that the same approximation factor can be
obtained by computing a Gomory-Hu tree and taking the $k$ lightest
cuts. The Gomory-Hu tree can be computed in $n$ maximum flow
computations, and the maximum flow can be computed deterministically
in
\begin{math}
  \apxO{m \min{m^{1/2}, n^{2/3}} \log U}
\end{math}
time for integer capacities between $1$ and $U$ \citep{gr-98}. This
gives a $\apxO{m n \min{m^{1/2}, n^{2/3}} \log U}$ deterministic time
$\parof{2 - 2/k}$-approximation for minimum $k$-cut, which is faster
than $\apxO{k m n}$ for sufficiently large $k$. (There are faster
randomized algorithms for maximum flow \citep{ls-14,madry-16}, but
these still lead to slower randomized running times than $\apxO{k m}$
for $k$-cut.) An LP based 2-approximation was derived by
\citet{nr-01}, and a combinatorial 2-approximation was given by
\citet{rs-08} (see also \citep{barahona-00}), but the running times
are worse than those implied by \citet{sv-95}. Finally, a
$\bigO{m n + n^2 \log n}$ deterministic running time was obtained by a
combinatorial algorithm by \citet{nk}, which is the best known
deterministic running time.  The constant factor of 2 is believed to
be essentially the best possible.  \citet{manurangsi-17} showed that
under the small set expansion hypothesis, for any fixed $\eps > 0$,
one cannot compute a $(2-\eps)$-approximation for the minimum $k$-cut
in $\poly{k,m,n}$ time unless $P = NP$.

The state of affairs for computing 2-approximate minimum $k$-cuts in
$\poly{k,m,n}$-time parallels the status of minimum cut. The fastest
randomized algorithm is an order of magnitude faster than the fastest
deterministic algorithm, while in the unweighted case the running
times are essentially equal. A basic question is whether there exists
a deterministic algorithm that computes a $2$-approximation in
$\apxO{k m}$ time, matching the randomized running time. As
\citeauthor{sv-95}'s algorithm reduces $k$-cut to $k$ minimum cuts,
the gap between the deterministic and randomized running times for
$k$-cut is not only similar to, but a reflection of, the gap between
the deterministic and randomized running times for minimum cut. An
$\apxO{m}$ deterministic algorithm for minimum cut would close the gap
for 2-approximate $k$-cut as well.

A second question asks if computing a 2-approximate $k$-cut is
qualitatively harder than computing the minimum cut. There is
currently a large gap between the fastest algorithm for minimum cut
and the fastest algorithm for 2-approximate minimum $k$-cut. Can one
compute 2-approximate minimum $k$-cuts as fast as minimum cuts -- in
$\apxO{m}$ randomized time?  Removing the linear dependence on $k$
would show that computing 2-approximate $k$-cuts is as easy as
computing a minimum cut.

\subsection{The main result}

We make progress on both of these questions with a deterministic and
nearly linear time $(2 + \eps)$-approximation scheme for minimum
$k$-cuts. To state the result formally, we first introduce an LP
relaxation for the minimum $k$-cut due to \citet{nr-01}.
\begin{align}
  \labelthisequation[L]{lp}
  \begin{aligned}
    \text{min } %
    & \sum_{e} \cost{e} x_e 
    \text{ over } %
    x: E \to \reals \\
    \text{s.t.\ } & \sum_{e \in T} x_e \geq k - 1 \text{ for all
      spanning trees } T,
    \\
    & %
    0 \leq x_e \leq 1 \text{ for all edges } e.
  \end{aligned}
\end{align}
The feasible integral solutions of the LP \refequation{lp} are
precisely the $k$-cuts in $G$. The main contribution of this work is a
nearly linear time approximation scheme for \refequation{lp}.

\begin{theorem}
  \labeltheorem{apx-lp} %
  In $\bigO{m \log^3(n) / \eps^2}$ deterministic time, one can compute
  an $\epspm$-multiplicative approximation to \refequation{lp}.
\end{theorem}

The integrality gap of \refequation{lp} is known to be
$\parof{2 - 2/n}$ \citep{nr-01,cgn-06}. Upon inspection, the
rounding algorithm can be implemented in $\bigO{m \log n}$ time,
giving the following nearly linear time
$\parof{2 + \eps}$-approximation scheme for $k$-cut.

\begin{theorem}
  \labeltheorem{apx-k-cut} %
  For sufficiently small $\eps > 0$, there is a deterministic
  algorithm that computes a $k$-cut with total capacity at most
  $\parof{2 + \eps}$ times the optimum value to \refequation{lp} in
  $\bigO{m \log^3(n) / \eps^2}$ time.
\end{theorem}

This algorithm should be compared with the aforementioned algorithms
of \citet{sv-95}, which computes a
$\parof{2 - \frac{2}{k}}$-approximation to the minimum $k$-cut in
$\apxO{k m}$ randomized time (with high probability), and of
\citet{nk}, which computes a $2$-approximate minimum cut in
\begin{math}
  \bigO{mn + n^2 \log n}
\end{math}
deterministic time. At the cost of a $\epsmore$-multiplicative factor,
we obtain a deterministic algorithm with nearly linear running time
for \emph{all values of $k$.} The approximation factor converges to
$2$, and we cannot expect to beat 2 under the small set expansion
hypothesis \citep{manurangsi-17}. Thus, \reftheorem{apx-k-cut} gives a
tight, deterministic, and nearly linear time approximation scheme for
$k$-cut. \reftheorem{apx-k-cut} leaves a little bit of room for
improvement: the hope is for a deterministic algorithm that computes
$\parof{2 - o(1)}$-approximate minimum $k$-cuts in $\apxO{m}$
time. Based on \reftheorem{apx-k-cut}, we conjecture that such an
algorithm exists.

\subsection{Overview of the algorithm}

We give a brief sketch of the algorithm, for the sake of informing
subsequent discussion on related work in \refsection{related-work}. A
more complete description of the algorithm begins in earnest in
\refsection{mwu}.

The algorithm consists of a nearly linear time approximation scheme
for the LP \refequation{lp}, and a nearly linear time rounding
scheme. The approximation scheme for solving the LP extends techniques
from recent work \citep{cq-17-soda}, applied to the dual of an
\emph{indirect reformulation} of \refequation{lp}. The rounding scheme
is a simplification of the rounding scheme by \citet{cgn-06}
for the more general Steiner $k$-cut problem, which builds on the
primal-dual framework of \citet{gw-95}.

The first step is to obtain a $\epsmore$-multiplicative approximation
to the LP \refequation{lp}. Here a $\epsmore$-multiplicative
approximation to the LP \refequation{lp} is a feasible vector $x$ of
cost at most a $\epsmore$-multiplicative factor greater than the
optimum value.

The LP \refequation{lp} can be solved exactly by the ellipsoid method,
with the separation oracle supplied by a minimum spanning tree (abbr.\
MST) computation, but the running time is a larger polynomial than
desired. From the perspective of fast approximations, the LP
\refequation{lp} is difficult to handle because it is an exponentially
large mixed packing and covering problem, with exponentially many
covering constraints alongside upper bounds on each edge. Fast
approximation algorithms for mixed packing and covering problems (e.g.
\citep{young-14,cq-18-soda}) give bicriteria approximations that
meet either the covering constraints or the packing constraints but
not both.  Even without consideration of the objective function, it is
not known how to find feasible points to general mixed packing and
covering problems in time faster than via exact LP
solvers. Alternatively, one may consider the dual of \refequation{lp},
as follows. Let $\trees$ denote the family of spanning trees in $G$.
\begin{align*}
  \text{max } %
  & (k-1)\sum_{T} y_T - \sum_{e \in E} z_e
    \text{ over } y: \trees \to \reals \text{ and } z: E \to \reals
  \\
  \text{s.t.\ }                 %
  & \sum_{T \ni e} y_T \leq \capacity{e} + z_e \text{ for all edges }
    e, \\
  & y_T \geq 0 \text{ for all spanning trees } T \in \trees, %
  \\
  & z_e \geq 0 \text{ for all edges } e.
\end{align*}
This program is not a positive linear program, and the edge potentials
$z \in \nnreals^E$ are difficult to handle by techniques such as
\citep{young-14,cq-17-soda,cq-18-soda}. It was not known, prior to
this work, how to obtain any approximation to \refequation{lp} (better
than its integrality gap) with running time faster than the ellipsoid
algorithm.

Critically, we consider the following larger LP instead of
\refequation{lp}. Let $\forests$ denote the family of all forests in
$G$.
\begin{align}
  \begin{aligned}
    \text{min } %
    & \sum_{e} \cost{e} x_e 
    \text{ over } %
    x: E \to \reals \\
    \text{s.t.\ } & \sum_{e \in T} x_e \geq \sizeof{F} + k - n \text{
      for all forests } F \in \forests,
    \\
    & %
    x_e \geq 0 \text{ for all edges } e \in E.
  \end{aligned}
      \labelthisequation[C]{kc}
\end{align}
\refequation{kc} is also an LP relaxation for $k$-cut. In fact,
\refequation{kc} is equivalent to \refequation{lp}, as one can verify
directly (see \reflemma{lp=kc} below). \refequation{kc} is obtained
from \refequation{lp} by adding all the knapsack covering constraints
\citep{cflp-00}, which makes the packing constraints ($x_e \leq 1$ for
each edge $e$) redundant.

Although the LP \refequation{kc} adds exponentially many constraints
to the original LP \refequation{lp}, \refequation{kc} has the
advantage of being a pure covering problem. Its dual is a pure packing
problem, as follows.
\begin{align}
  \begin{aligned}
    \text{maximize } & \sum_{F \in \forests} \parof{\sizeof{F} + k -
      n} y_F %
    \text{ over } y: \forests \to \reals %
    \\
    \text{ s.t.\ } & \sum_{F \ni e} y_F \leq \capacity{e} \text{ for
      all
      edges } e \in E,\\
    & y_F \geq 0 \text{ for all forests } F \in \forests.
  \end{aligned}
      \labelthisequation[P]{packing-forests}
\end{align}
The above LP packs forests into the capacitated graph $G$ where the
value of a forest $F$ depends on the number of edges it contains,
$\sizeof{F}$. The objective value $\sizeof{F} + k - n$ of a forest $F$
is a lower bound on the number of edges $F$ contributes to any
$k$-cut.  Clearly, we need only consider forests with at least
$n - k + 1$ edges.

To approximate the desired LP \refequation{kc}, we apply the MWU
framework to the above LP \refequation{packing-forests}, which
generates $\epspm$-multiplicative approximations to both
\refequation{packing-forests} and its dual, \refequation{kc}.
Implementing the MWU framework in nearly linear time is not immediate,
despite precedent for similar problems.  In the special case where
$k = 2$, the above LP \refequation{packing-forests} fractionally packs
spanning trees into $G$.  A nearly linear time approximation scheme
for $k = 2$ is given in previous work \citep{cq-17-soda}. The general
case with $k > 2$ is more difficult for two reasons. First, the family
of forests that we pack is larger than the family of spanning
trees. Second, \refequation{packing-forests} is a weighted packing
problem, where the coefficients in the objective depends on the number
of edges in the forest. When $k = 2$, we need only consider spanning
trees with $n - 1$ edges, so all the coefficients are 1 and the
packing problem is unweighted.  The heterogeneous coefficients in the
objective create technical complications in the MWU framework, as the
Lagrangian relaxation generated by the framework is no longer solved
by a MST. In \refsection{mwu}, we give an overview of the MWU
framework and discuss the algorithmic complications in greater depth.
In \refsection{mwu-oracle} and \refsection{mwu-weight-update}, we show
how to extend the techniques of \citep{cq-17-soda} with some new
observations to overcome these challenges and approximate the LP
\refequation{packing-forests} in nearly the same time as one can
approximately pack spanning trees. Ultimately, we obtain the following
deterministic algorithm for approximating the LP
\refequation{packing-forests}.

\begin{theorem}
  \labeltheorem{apx-lp} %
  In $\bigO{m \log^3(n) / \eps^2}$ deterministic time, one can compute
  $\epspm$-multiplicative approximations to
  \refequation{packing-forests}, \refequation{kc} and
  \refequation{lp}.
\end{theorem}

The second step, after computing a fractional solution $x$ to
\refequation{lp} with \reftheorem{apx-lp}, is to round $x$ to a
discrete $k$-cut. The rounding step is essentially that of
\citet{cgn-06} for Steiner $k$-cuts. Their case is more general
than ours; we simplify their rounding scheme, and pay greater
attention to the running time. The rounding scheme is based on the
elegant primal-dual MST algorithm of \citet{gw-95}.
\begin{theorem}
  \labeltheorem{rounding} Given a feasible solution $x$ to
  \refequation{kc}, one can compute a $k$-cut $C$ with cost at most
  $2 \parof{1 - 1/n}$ times the cost of $x$ in $\bigO{m \log n}$ time.
\end{theorem}

Applying \reftheorem{rounding} to the output of \reftheorem{apx-lp}
gives \reftheorem{apx-k-cut}.

Computing the minimum $k$-cut via the LP \refequation{lp} has
additional benefits. First, computing a minimum $k$-cut with an
approximation factor relative to the LP may be much stronger than the
same approximation factor relative to the original problem, as LP's
perform well in practice. Second, the solution to the LP gives a
certificate of approximation ratio, as we can compare the rounded
$k$-cut to the LP solution to infer an upper bound on the
approximation ratio that may be smaller than 2.

Lastly, we note that some data structures can be simplified at the
cost of randomization by using a randomized MWU framework instead
\citep{cq-18-soda}. These modifications are discussed at the end of
\refsection{mwu-weight-update}.

\subsection{Further related results and discussion}
\labelsection{related-work}

\paragraph{Fixed parameter tractability.}
The $k$-cut results reviewed above were focused on either exact
polynomial-time algorithms for constant $k$ or
$\parof{2 - o(1)}$-approximations in $\poly{k,m,n}$ time, with a
particular emphasis on the fastest algorithms in the $\poly{k,m,n}$
time regime. There is also a body of literature concerning fixed
parameter tractable algorithms for $k$-cut. \citet{defpr-03} showed
that $k$-cut is $W[1]$-hard in $k$ even for simple unweighted graphs;
$W[1]$-hardness implies that it is unlikely to obtain a running time
of the form $f(k) \poly{m,n}$ for any function $f$. On the other hand,
\citet{kt-11} showed that $k$-cut is fixed parameter tractable in the
number of edges in the cut. More precisely, \citet{kt-11} gave a
deterministic algorithm that, for a given cardinality
$s \in \naturalnumbers$, time, either finds a $k$-cut with at most $s$
edges or reports that no such cut exists in
\begin{math}
  \bigO{s^{s^{O(s)}} n^2}
\end{math}
time. The running time was improved to
\begin{math}
  \bigO{2^{\bigO{s^2 \log s}}n^4 \log n}
\end{math}
deterministic time and
\begin{math}
  \apxO{2^{\bigO{s} \log k} n^2}
\end{math}
randomized time by \citet{cchpp-16}.

Besides exact parameterized algorithms for $k$-cut, there is interest
in approximation ratios between 1 and 2.  \citet{xcy-11} showed that
by adjusting the reduction of \citet{sv-95} to use (exact)
minimum $\ell$-cuts -- instead of minimum (2-)cuts, for any choice of
$\ell \in \setof{2,\dots,k-1}$ -- one can obtain
\begin{math}
  \parof{2 - \frac{\ell}{k} + \bigO{\frac{\ell^2}{k^2}}}
\end{math}-approximate minimum $k$-cuts in
\begin{math}
  n^{\bigO{\ell}}
\end{math}
time.  For sufficiently small $\eps > 0$, by setting
$\ell \approx \eps k$, this gives a
\begin{math}
  n^{\bigO{\eps k}}
\end{math}
time algorithm for $\parof{2 - \eps}$-approximate minimum cuts.

The hardness results of \citet{defpr-03} and \citet{manurangsi-17} do
not rule out approximation algorithms with approximation ratio better
than 2 and running times of the form $f(k)\poly{m,n}$ (for any
function $f$). Recently, \citet{gll-18} gave a FPT algorithm that,
for a particular constant $c \in (0,1)$, computes a
$(2 - c)$-approximate $k$-cut in
\begin{math}
  2^{\bigO{k^6}} \apxO{n^4}
\end{math}
time. This improves the running time of \citep{xcy-11} for $\eps = c$
and $k$ greater than some constant. Further improvements by
\citet{gll-18-focs} achieved a deterministic 1.81 approximation in
$2^{\bigO{k^2}} n^{\bigO{1}}$ time, and a randomized
$\epsmore$-approximation (for any $\eps > 0$) in
$\parof{k / \eps}^{\bigO{k}} n^{k + \bigO{1}}$ time.

\paragraph{Knapsack covering constraints.}
We were surprised to discover that adding the knapsack covering
constraints allowed for \emph{faster} approximation
algorithms. Knapsack covering constraints, proposed by \citet{cflp-00}
in the context of capacitated network design problems, generate a
stronger LP whose solutions can be rounded to obtain better
approximation factors
\citep{cflp-00,ky-01,ky-05,cckk-15}. However, the larger LP
can be much more complicated and is usually more difficult to
solve. We recently obtained a faster approximation scheme for
approximating covering integer programs via knapsack covering
constraints, but the dependency on $\eps$ is much worse, and the
algorithm has a ``weakly nearly linear'' running time that suffers
from a logarithmic dependency on the multiplicative range of input
coefficients \citep{cq-19}.

\paragraph{Fast approximations via LP's.}
Linear programs have long been used to obtain more accurate
approximations to NP-Hard problems. Recently, we have explored the use
of fast LP solvers to obtain \emph{faster} approximations, including
situations where polynomial time algorithms are known. In recent work
\citep{cq-18-christofides}, we used a linear time approximation to an
LP relaxation for Metric TSP (obtained in \citep{cq-17-focs}) to
effectively sparsify the input and accelerate Christofides'
algorithm. While nearly linear time approximations for complicated
LP's are surprising in and of itself, perhaps the application to
obtain faster approximations for combinatorial problems is more
compelling. We think this result is an important data point for this
approach.

\section{Reviewing the MWU framework and identifying bottlenecks}

\labelsection{mwu}

In this section, let $\eps > 0$ be fixed. It suffices to assume that
$\eps \geq 1 / \poly{n}$, since below this point one can use the
ellipsoid algorithm instead and still meet the desired running
time. For ease of exposition, we seek only a $\apxmore$-multiplicative
approximation; a $\epsmore$-multiplicative approximation with the same
asymptotic running time follows by decreasing $\eps$ by a constant
factor.

\subsection{$k$-cuts as a (pure) covering problem}
As discussed above, the first (and most decisive) step towards a fast,
fractional approximation to $k$-cut is identifying the right LP. The
standard LP \refequation{lp} is difficult because it is a mixed
packing and covering problem, and fast approximation algorithms for
mixed packing and covering problems lead to bicriteria approximations
that we do not know how to round. On the other hand, extending
\refequation{lp} with all the knapsack cover constraints makes the
packing constraints $x_e \leq 1$ redundant (as shown below), leaving
the pure covering problem, \refequation{kc}. The LPs \refequation{lp}
and \refequation{kc} have essentially equivalent solutions in the
following sense.
\begin{lemma}
  \labellemma{lp=kc} Any feasible solution $x \in \nnreals^E$ to
  \refequation{lp} is a feasible solution to \refequation{kc}. For any
  feasible solution $x$ to \refequation{kc}, the truncation
  $x' \in \nnreals^n$ defined by $x_e' = \max{x_e,1}$ is a feasible
  solution to both \refequation{lp} and \refequation{kc}.
\end{lemma}
\begin{proof}
  Let $x$ be a feasible solution to \refequation{lp}. We claim that
  $x$ is feasible in \refequation{kc}. Indeed, let $F$ be a forest,
  and extend $F$ to a tree $T$. Then
  \begin{align*}
    \sum_{e \in F} x_e %
    = %
    \sum_{e \in T} x_e - \sum_{e \in T \setminus F} x_e %
    \geq                                                %
    k - 1 - \sizeof{T \setminus F}                      %
    =                                                   %
    k - 1 - \parof{n-1 - \sizeof{F}}                    %
    =                                                   %
    \sizeof{F} + k - n,
  \end{align*}
  as desired.

  Conversely, let $x \in \nnreals^E$ be a feasible solution to
  \refequation{kc}, and let $x'$ be the coordinatewise maximum of $x$
  and $\ones$. Since $x' \leq \ones$, and the covering constraints in
  \refequation{lp} are a subset of the covering constraints in
  \refequation{kc}, if $x'$ is feasible in \refequation{kc} then it is
  also feasible in \refequation{lp}. To show that $x'$ is feasible in
  \refequation{kc}, let $F$ be a forest. Let
  $F' = \setof{e \in F \where x_e > 1}$ be the edges in $F$ truncated
  by $x'$ and let $F'' = F \setminus F'$ be the remaining edges. Since
  \tagr $x'_e = 1$ for all $e \in F'$ and $x_e' = x_e$ for all
  $e \in F''$, and \tagr $x$ covers $F''$ in \refequation{kc}, we have
  \begin{align*}
    \sum_{e \in F} x_e'          %
    =                           %
    \sum_{e \in F'} x_e' + \sum_{e \in F''} x'_e %
    \tago{=}
    \sizeof{F'} + \sum_{e \in F''} x_e %
    \tago{\geq}                                  %
    \sizeof{F'} + \sizeof{F''} + k - n    %
    =                                     %
    \sizeof{F} + k - n,
  \end{align*}
  as desired.
\end{proof}

While having many more constraints than \refequation{lp},
\refequation{kc} is a pure covering problem, for which finding a
feasible point (faster than an exact LP solver) is at least
plausible. The dual of \refequation{kc} is the LP
\refequation{packing-forests}, which packs forests in the graph and
weights each forest by the number of edges minus $\parof{n - k}$. The
coefficient of a forest in the objective of
\refequation{packing-forests} can be interpreted as the number of
edges that forest must contribute to any $k$-cut.

\subsection{A brief sketch of width-independent MWU}

We apply a \emph{width-independent} version of the MWU framework to
the packing LP \refequation{packing-forests}, developed by
\citet{gk-07} for multicommodity flow problems and generalized
by \citet{young-01}. We restrict ourselves to a sketch of the
framework and refer to previous work for further details.

The width-independent MWU framework is a monotonic and iterative
algorithm that starts with an empty solution $y = \zeroes$ to the LP
\refequation{packing-forests} and increases $y$ along forests that
solve certain Lagrangian relaxations to
\refequation{packing-forests}. Each Lagrangian relaxation is designed
to steer $y$ away from packing forests that have edges that are
already tightly packed. For each edge $e$, the framework maintains a
weight $\weight{e}$ that (approximately) exponentiates the load of
edge $e$ w/r/t the current forest packing $y$, as follows:
\begin{align*}
  \ln{\capacity{e} \weight{e}}  %
  \approx %
  \frac{\log n}{\eps} \frac{\sum_{F \ni e} y_F}{\capacity{e}}. %
  \labelthisequation{weight-load}
\end{align*}
The weight can be interpreted as follows.  For an edge $e$, the value
$\frac{\sum_{F \ni e} y_F}{\capacity{e}}$ is the amount of capacity
used by the current packing $y$ relative to the capacity of the edge
$e$. We call $\frac{\sum_{F \ni e} y_F}{\capacity{e}}$ the (relative)
\emph{load} on edge $e$ and is $\leq 1$ if $y$ is a feasible
packing. The weight $\weight{e}$ is exponential in the load on the
edge, where the exponential is amplified by the leading coefficient
$\frac{\log n}{\eps}$.  Initially, the empty solution $y = \emptyset$
induces zero load on any edge and we have
$\weight{e} = \frac{1}{\capacity{e}}$ for each edge $e$.

Each iteration, the framework solves the following Lagrangian
relaxation of \refequation{packing-forests}:
\begin{align*}
  \text{maximize }              %
  & \sum_{F \in \forests} \parof{\sizeof{F} + k - n}
    z_F \text{ over } z: \forests \to \nnreals %
    \text{ s.t.\ }
    \sum_{e \in E} \weight{e} \sum_{F \ni e} z_F %
    \leq                                         %
    \sum_{e \in E} \weight{e} \capacity{e}.
    \labelthisequation[R]{relaxation}
\end{align*}
Given a $\apxmore$-approximate solution $z$ to the above, the
framework adds $\delta z$ to $y$ for a carefully chosen value
$\delta > 0$ (discussed in greater detail below). The next iteration
encounters a different relaxation, where the edge weights $\weight{e}$
are increased to account for the loads increased by adding $\delta
z$. Note that the edge weights $\weight{e}$ are monotonically
increasing over the course of the algorithm.

At the end of the algorithm, standard analysis shows that the
fractional forest packing $y$ has objective value $\apxless \opt$, and
that $\apxless y$ satisfies all of the packing constraints. The error
can be made one-sided by scaling $y$ up or down. Moreover, it can be
shown that at some point in the algorithm, an easily computable
rescaling of $w$ is a $\apxmore$-relative approximation for the
desired LP \refequation{kc} (see for example
\citep{gk-98,gk-07,cq-17-focs}). Thus, although we may appear more
interested in solving the dual LP \refequation{packing-forests}, we
are approximating the desired LP \refequation{kc} as well.

The choice of $\delta$ differentiates this ``width-independent'' MWU
from other MWU-type algorithms in the literature. The step size
$\delta$ is chosen small enough that no weight increases by more than
an $\exp{\eps}$-multiplicative factor, and large enough that some
weight increases by (about) an $\exp{\eps}$-multiplicative factor. The
analysis of the MWU framework reveals that
\begin{math}
  \rip{\weight}{\capacity} \leq n^{\bigO{1/\eps}}
\end{math}
at all times. In particular, each weight can increase by an
$\exp{\eps}$-multiplicative factor at most
$\bigO{\frac{\ln n}{\eps^2}}$ times, so there are at most
\begin{math}
  \bigO{\frac{n \ln n}{\eps^2}}
\end{math}
iterations total.

\subsection{Two bottlenecks}
\labelsection{bottlenecks}

The MWU framework alternates between (a) solving the relaxation
\refequation{relaxation} induced by edge weights $\weight{e}$ and (b)
updating the weights $\weight{e}$ for each edge in response to the
solution to the relaxation. As the framework requires
$\bigO{\frac{m \log n}{\eps^2}}$ iterations, both parts must be
implemented in polylogarithmic amortized time to reach the desired
running time. A sublinear per-iteration running time seems unlikely by
the following simple observations.

Consider first the complexity of simply expressing a solution.  Any
solution $z$ to \refequation{relaxation} is indexed by forests in
$G$. A forest can have $\bigOmega{n}$ edges and requires
$\bigOmega{n \log n}$ bits to specify. Writing down the index of just
one forest in each of $\bigO{\frac{m \log n}{\eps^2}}$ iterations
takes $\bigO{\frac{m n \log^2 n}{\eps^2}}$ time.  The difficulty of
even writing down a solution to \refequation{packing-forests} is not
just a feature of the MWU framework. In general, there exists an
optimal solution to \refequation{packing-forests} that is supported by
at most $m$ forests, as $m$ is the rank of the implicit packing
matrix. Writing down $m$ forests also requires $\bigOmega{m n \log n}$
bits.  Thus, either on a per-iteration basis in the MWU framework or
w/r/t to the entire LP, the complexity of the output suggests a
quadratic lower bound on the running time.

A second type of bottleneck arises from updating the weights. The
weights $\weight{e}$ for each edge reflect the load induced by the
packing $y$, per the formula \refequation{weight-load}. After
computing a solution $z$ to the relaxation \refequation{relaxation},
and updating $y \gets y + \delta z$, we need to update the weights
$\weight{e}$ to reflect the increased load from $\delta z$. In the
worst case, $\delta z$ packs into every edge, requiring us to update
$\bigO{m}$ individual weights. At the very least, $\delta z$ should
pack into the edges of at least one forest, and thus effect
$\bigOmega{n}$ edges. Updating $n$ edge weights in each iteration
requires $\bigO{\frac{m n \log n}{\eps^2}}$ time.

Even in hindsight, implementing either part -- solving the relaxation
or updating the weights -- in isolation in sublinear time remains
difficult. Our algorithm carefully plays both parts off each other, as
co-routines, and amortizes against invariants revealed by the analysis
of the MWU framework. The seemingly necessary dependence between parts
is an important theme of this work and an ongoing theme from previous
work \citep{cq-17-soda,cq-17-focs,cq-19}.

\section{Greedily finding forests to pack in $\bigO{\log^2 n}$
  amortized time}

\labelsection{mwu-oracle}

The MWU framework reduces \refequation{packing-forests} to a sequence
of problems of the form \refequation{relaxation}.  An important aspect
of the Lagrangian approach is that satisfying the single packing
constraint in \refequation{relaxation} is much simpler than
simultaneously satisfying all of the packing constraints in
\refequation{packing-forests}. With only 1 packing constraint, it
suffices to (approximately) identify the best bang-for-buck forest $F$
and taking as much as can fit in the packing constraint. The
``bang-for-buck'' ratio of a forest $F$ is the ratio
\begin{align*}
  \frac{\sizeof{F} + k - n}{\sum_{e \in F} \weight{e}},
\end{align*}
where $\sizeof{F}$ is the number of edges in $F$. Given a forest $F$
(approximately) maximizing the above ratio, we set $z = \gamma e_F$
for $\gamma$ as large as possible as fits in the single packing
constraint. Note that, when $k = 2$, the optimal forest is the minimum
weight spanning tree w/r/t $\weight$.

We first consider the simpler problem of maximizing the above ratio
over forests $F$ with exactly $\sizeof{F} = \ell$ edges, for some
$\ell > n - k$. Recall that the MST can be computed greedily by
repeatedly adding the minimum weight edge that does not induce a
cycle. Optimality of the greedy algorithm follows from the fact that
spanning trees are the bases of a matroid called the \emph{graphic
  matroid}. The forests of exactly $\ell$ edges are also the bases of
a matroid; namely, the restriction of the graphic matroid to forests
of at most $\ell$ edges. In particular, the same greedy procedure
computes the minimum weight forest of $\ell$ edges. Repeating the
greedy algorithm for each choice of $\ell$, one can solve
\refequation{relaxation} in $\bigO{k m \log n}$ time for each of
\begin{math}
  \bigO{\frac{m \log n}{\eps^2}}
\end{math}
iterations.

Stepping back, we want to compute the minimum weight forest with
$\ell$ edges for a range of $k - 1$ values of $\ell$, and we can run
the greedy algorithm for each choice of $\ell$. We observe that the
greedy algorithm is oblivious to the parameter $\ell$, except for
deciding when to stop. We can run the greedy algorithm \emph{once} to
build the MST, and then simulate the greedy algorithm for any value of
$\ell$ by taking the first $\ell$ edges added to the MST.

\begin{lemma}
  \labellemma{minimum-weight-sub-forest} Let $T$ be the minimum weight
  spanning tree w/r/t $\weight$. For any $\ell \in [n-1]$, the minimum
  weight forest w/r/t $\weight$ with $\ell$ edges consists of the
  first $\ell$ minimum weight edges of $T$.
\end{lemma}

\reflemma{minimum-weight-sub-forest} effectively reduces
\refequation{relaxation} to one MST computation, which takes
$\bigO{m\log n}$ time. Repeated over $\bigO{\frac{m \log n}{\eps^2}}$
iterations, this leads to a $\bigO{\frac{m^2 \log^2 n}{\eps^2}}$
running time.  As observed previously \citep{thorup-08,cq-17-soda},
the minimum weight spanning tree does not have to be rebuilt from
scratch from one iteration to another, but rather adjusted dynamically
as the weights change.

\begin{lemma}[{\citet{hlt-01}}]
  \labellemma{dynamic-mst} In $\bigO{\log^2 n}$ amortized time per
  increment to $\weight$, one can maintain the MST w/r/t
  $\weight$.
\end{lemma}

The running time of \reflemma{dynamic-mst} depends on the number of
times the edge weights change, so we want to limit the number of
weight updates exposed to \reflemma{dynamic-mst}. It is easy to see
that solving \refequation{relaxation} w/r/t a second set of weights
$\apxweight$ that is a $\epspm$-multiplicative factor coordinatewise
approximation of $\weight$ gives a solution that is a
$\epspm$-multiplicative approximation to \refequation{relaxation}
w/r/t $\weight$. We maintain the MST w/r/t an approximation
$\apxweight$ of $\weight$, and only propagate changes from $\weight$
to $\apxweight$ when $\weight$ is greater than $\apxweight$ by at
least a $\epsmore$-multiplicative factor. As mentioned in
\refsection{mwu}, a weight $\weight{e}$ increases by a
$\epsmore$-multiplicative factor at most
$\bigO{\frac{\log n}{\eps^2}}$ times. Applying \reflemma{dynamic-mst}
to the discretized weights $\apxweight$ and amortizing against the
total growth of weights in the system gives us the following.
\begin{lemma}
  In $\bigO{\frac{m \log^3 n}{\eps^2}}$ total time, one can maintain
  the MST w/r/t a set of weights $\apxweight$ such that, for all
  $e \in E$, we have $\apxweight{e} \in \epspm \weight{e}$. Moreover,
  the MST makes at most $\bigO{\frac{m \log n}{\eps^2}}$ edge updates
  total.
\end{lemma}

Given such an MST $T$ as above, and
\begin{math}
  \ell \in \setof{n - k+1,\dots,n-1}
\end{math}
we need the $\ell$ minimum ($\apxweight$-)weight edges to form an
(approximately) minimum weight forest $F$ of $\ell$ edges. However,
the data structure of \reflemma{dynamic-mst} does not provide a list
of edges in increasing order of weight. We maintain the edges in
sorted order separately, where each time the dynamic MST replaces one
edge with another, we make the same update in the sorted
list. Clearly, such a list can be maintained in $\bigO{\log n}$ time
per update by dynamic trees. Our setting is simpler because the range
of possible values of any weight $\apxweight{e}$ is known in advance
as follows.  For $e \in E$, define
\begin{align*}
  \weights{e} =                 %
  \setof{                       %
  \frac{\epsmore^i}{\cost{e}}   %
  \where                        %
  i \in
  \setof{0,1,\dots,\bigO{\frac{\log n}{\eps^2}}} %
  }.
\end{align*}
Then $\apxweight{e} \in \weights{e}$ for all $e \in E$ at all
times. Define
\begin{math}
  \leaves = \setof{ (e,\alpha) \where \alpha \in \weights{e}}.
\end{math}
The set $\leaves$ represents the set of all possible assignments of
weights to edges that may arise. As mentioned above,
\begin{math}
  \sizeof{\leaves} = \bigO{\frac{m \log n}{\eps^2}}.
\end{math}
Let $B$ be a balanced binary tree over $\leaves$, where $\leaves$ is
sorted by increasing order of the second coordinate $\apxweight{e}$
(and ties are broken arbitrarily). The tree $B$ has height
$\log \sizeof{\leaves} = \bigO{\log m}$ and can be built in
$\bigO{\sizeof{\leaves}} = \bigO{\frac{m \log n}{\eps^2}}$
time\footnote{or even less time if we build $B$ lazily, but
  constructing $B$ is not a bottleneck}.

We mark the leaves based on the edges in $T$.  Every time the MST $T$
adds an edge $e$ of weight $\apxweight{e}$, we mark the corresponding
leaf $(e,\apxweight{e})$ as marked. When an edge $e$ is deleted, we
unmark the corresponding leaf. Lastly, when an edge $e \in T$ has its
weight increased from $\alpha$ to $\alpha'$, we unmark the leaf
$(e,\alpha)$ and mark the leaf $(e,\alpha')$. Note that only edges in
$T$ have their weight changed.

For each subtree of $T$, we track aggregate information and maintain
data structures over the set of all marked leaves in the subtree. For
a node $b$ in $B$, let $\leaves{b}$ be the set of marked leaves in the
subtree rooted at $b$. For each $b \in B$ we maintain two quantities:
(a) the number of leaves marked in the subtree rooted at $b$,
$\sizeof{\leaves{b}}$; and (b) the sum of edges weights of leaves
marked in the subtree rooted at $b$,
$W_b = \sum_{(e,\alpha) \in \leaves{b}} \alpha$. Since the height of
$B$ is $\bigO{\log n}$, both of these quantities can be maintained in
$\bigO{\log n}$ time per weight update.

\begin{lemma}
  \labellemma{ordered-mst} %
  In $\bigO{m \log{n} / \eps^2}$ time initially and $\bigO{\log n}$
  time per weight update, one can maintain a data structure that,
  given $\ell \in [n-1]$, returns in $\bigO{\log n}$ time (a) the
  $\ell$ minimum weight edges of the MST (implicitly), and (b) the
  total weight of the first $\ell$ edges of the MST.
\end{lemma}

With \reflemma{dynamic-mst} and \reflemma{ordered-mst}, we can now
compute, for any $\ell \in [n-1]$, a $\epsmore$-approximation to the
minimum weight forest of $\ell$ edges, along with the sum weight of
the forest, both in logarithmic time. To find the best forest, then,
we need only query the data structure for each of the $k-1$ integer
values from $n - k + 1$ to $n - 1$. That is, excluding the time to
maintain the data structures, we can now solve the relaxation
\refequation{relaxation} in $\bigO{k \log n}$ time per iteration.

At this point, we still require $\bigO{k m \log n}$ time just to solve
the Lagrangian relaxations \refequation{relaxation} generated by the
MWU framework. (There are other bottlenecks, such as updating the
weights at each iteration, that we have not yet addressed.) To remove
the factor of $k$ in solving \refequation{relaxation}, we require one
final observation.

\begin{lemma}
  \labellemma{greedy-search} The minimum ratio subforest of $T$ can be
  found by binary search.
\end{lemma}
\begin{proof}%
  \newcommand{\sumweight}[1]{\sum_{j=1}^{n - k + #1} \apxweight{j}}%
  Enumerate the MST edges $e_1,\dots,e_{n-1} \in T$ in increasing
  order of weight.  For ease of notation, we denote
  $\apxweight{i} = \apxweight{e_i}$ for $i \in [n-1]$. We define a
  function $f: [k-1] \to \preals$ by
  \begin{align*}
    f(i) = \frac{i}{\sumweight{i}}.
  \end{align*}
  For each $i$, $f(i)$ is the ratio achieved by the first $n - k + i$
  edges of the MST. Our goal is to maximize $f(i)$ over $i \in
  [k-1]$. For any $i \in [k-2]$, we have
  \begin{align*}
    f(i+1) - f(i)               %
    &=
      \frac{i+1}{\sumweight{i + 1}} -
      \frac{i}{\sumweight{i}} %
    \\
    &=                           %
      \frac{                     %
      (i+1) \sumweight{i} - i \sumweight{i + 1}                        %
      }{                        %
      \parof{\sumweight{i+1}}\parof{\sumweight{i}}                    %
      } %
          =
          \frac{\sumweight{i}- i \apxweight{n - k + i + 1}}{\parof{\sumweight{i+1}}\parof{\sumweight{i}}}.
  \end{align*}
  Since the denominator $\parof{\sumweight{i+1}}\parof{\sumweight{i}}$
  is positive, we have
  \begin{math}
    f(i+1) \leq f(i) \iff i \apxweight{n - k + i + 1} \geq
    \sumweight{i}.
  \end{math}
  If $i \apxweight{n - k + i + 1} \geq \sumweight{i}$ for all
  $i \in [k-2]$, then $f(1) < f(2) < \cdots < f(k-1)$, so $f(i)$ is
  maximized by $i = k - 1$. Otherwise, let $i_0 \in [k-2]$ be the
  first value of $i$ such that $f(i+1) \leq f(i)$. For $i_1 \geq i_0$,
  as \tagr $\apxweight{e_i}$ is increasing in $i$, \tagr
  $f(i_0 + 1) \leq f(i_0)$, we have
  \begin{align*}
    \sumweight{i_1} - i_1 w_{n - k + i_1 + 1} %
    &=                                                     %
      \sumweight{i_0} - i_0 \apxweight{n-k+i_1 + 1} + \sum_{j=n-k+i_0 + 1}^{n
      - k + i_1} \apxweight{j} - (i_1 - i_0) \apxweight{n-k+i_1+1} %
    \\
    &\tago{\leq}                          %
      \sumweight{i_0} - i_0 \apxweight{n - k + i_0 + 1} %
      \tago{\leq}                      %
      0.
  \end{align*}
  That is, $f(i_1 + 1) \leq f(i_1)$ for all $i_1 \geq i_0$. Thus $f$
  consists of one increasing subsequence followed by a decreasing
  subsequence, and its global maximum is the unique local maximum.
\end{proof}

By \reflemma{greedy-search}, the choice of $\ell$ can be found by
calculating the ratio of $\bigO{\log k}$ candidate forests. By
\reflemma{ordered-mst}, the ratio of a candidate forest can be
computed in $\bigO{\log n}$ time.

\begin{lemma}
  \labellemma{polylog-oracle} Given the data structure of
  \reflemma{ordered-mst}, one can compute a $\apxmore$-multiplicative
  approximation to \refequation{relaxation} in $\bigO{\log n \log k}$
  time.
\end{lemma}

The polylogarithmic running time in \reflemma{polylog-oracle} is
surprising when considering that solutions to \refequation{relaxation}
should require at least a linear number of bits, as discussed earlier
in \refsection{bottlenecks}. In hindsight, a combination of additional
structure provided by the MWU framework and the LP
\refequation{packing-forests} allows us to apply data structures that
effectively compress the forests and output each forest in
polylogarithmic amortized time. Implicit compression of this sort also
appears in previous work \citep{cq-17-soda,cq-17-focs,cq-19}.

\section{Packing greedy forests in $\bigO{\log^2 n}$ amortized time}

\labelsection{mwu-weight-update}

In \refsection{mwu-oracle}, we showed how to solve
\refequation{relaxation} in polylogarithmic time per iteration.  In
this section, we address the second main bottleneck: updating the
weights $w$ after increasing $y$ to $y + \delta z$ per the formula
\refequation{weight-load}, where $z$ is an approximate solution to the
relaxation \refequation{relaxation} and $\delta > 0$ is the largest
possible value such that no weight increases by more than a
$\epsmore$-multiplicative factor. As discussed in
\refsection{bottlenecks}, this may be hard to do in polylogarithmic
time when many of the edges $e \in E$ require updating.

A sublinear time weight update must depend heavily on the structure of
the solutions generated to \refequation{relaxation}. In our case, each
solution $z$ to a relaxation \refequation{relaxation} is of the form
$\gamma e_F$, where $e_F$ is the indicator vector of a forest $F$ and
$\gamma > 0$ is a scalar as large as possible subject to the packing
constraint in \refequation{relaxation}. We need to update the weights
to reflect the loads induced by $\delta z = \delta \gamma e_F$, where
$\delta$ is chosen large as possible so that no weight increases by
more than an $\exp{\eps}$-multiplicative factor. With this choice of
$\delta$, the weight update simplifies to the following formula. Let
$\weight$ denote the set of weights before the updates and $\weight'$
denote the set of weights after the updates. For a solution
$z = \gamma e_F$, we have
\begin{align*}
  \weight{e}' =                 %
  \begin{cases}
    \weight{e} &\text{if } e \notin F, \\
    \exp{ %
      \frac{\eps \min_{f \in F} \capacity{f}}{\capacity{e}} 
    }%
    &\text{if } e \in F.
  \end{cases}
      \labelthisequation{bottleneck-weight-update} %
\end{align*}
The weight update formula above can be interpreted as follows.
Because our solution is supported along a single forest $F$, the only
edges whose loads are effected are those in the forest $F$. As load is
relative to the capacity of an edge $e$, the increase of the logarithm
the weight $\weight{e}$ of an edge $e \in F$ is inversely proportional
to its capacity. By choice of $\delta$, the minimum capacity edge
$\argmin_{f \in F} \capacity{f}$ has its weight increased by an
$\exp{\eps}$ multiplicative factor. The remaining edges with larger
capacity each have the logarithm of their weight increased in
proportion to the ratio of the bottleneck capacity to its own
capacity.

Simplifying the weight update formula does not address the basic
problem of updating the weights of every edge in a forest $F$,
\emph{without visiting every edge in $F$}. Here we require
substantially more structure as to how the edges in $F$ are
selected. We observe that although there may be $\bigOmega{n}$ edges
in $F$, we can always decompose $F$ into a logarithmic number of
``canonical subforests'', as follows.

\begin{lemma}
  \labellemma{canonical-subforests} One can maintain, in
  $\bigO{\log n}$ time per update to the MST $T$, a collection of
  subforests $\subforests$ such that:
  \begin{mathresults}
  \item $\sizeof{\subforests} = \bigO{n \log n}$.
  \item Each edge $e \in T$ is contained in $\bigO{\log n}$ forests.
  \item For each $\ell \in [n-1]$, the forest $F$ consists of the
    $\ell$ minimum weight edges in $T$ decomposes uniquely into the
    disjoint union of $\bigO{\log n}$ forests in $\forests$. The
    decomposition can be computed in $\bigO{\log n}$ time.
  \end{mathresults}
\end{lemma}

In fact, the collection of subforests is already maintained implicitly
in \reflemma{ordered-mst}. Recall, from \refsection{mwu-oracle}, the
balanced binary tree $B$ over the leaf set $\leaves$, which consists
of all possible discretized weight-to-edge assignments and is ordered
in increasing order of weight. Leaves are marked according to the
edges in the MST $T$, and each node is identified with the forest
consisting of all marked leaves in the subtree rooted at the node. For
each $\ell \in [n-1]$, the forest $F_{\ell}$ induced by the $\ell$
minimum weight edges in $T$ is the set of marked leaves over an
interval of $\leaves$. The interval decomposes into the disjoint union
of leaves of $\bigO{\log n}$ subtrees, which corresponds to
decomposing $F_{\ell}$ into the disjoint union of marked leaves of
$\bigO{\log n}$ subtrees of $B$. That is, the forests of marked leaves
induced by subtrees of $B$ gives the ``canonical forests''
$\subforests$ that we seek.

The following technique of decomposing weight updates is critical to
previous work \citep{cq-17-soda,cq-17-focs,cq-19}; we briefly
discuss the high-level ideas and refer to previous work for complete
details.

Decomposing the solution into a small number of known static sets is
important because weight updates can be simulated over a \emph{fixed
  set} efficiently. The data structure \refroutine{lazy-inc}, defined
in \citep{cq-17-soda} and inspired by techniques by \citet{young-14},
simulates a weight update over a fixed set of weights in such a way
that the time can be amortized against the logarithm of the increase
in each of the weights. As discussed above, the total logarithmic
increase in each of the weights is bounded from above. The data
structure \refroutine{lazy-inc} is dynamic, allowing insertion and
deletion into the underlying set, in $\bigO{\log n}$ time per
insertion or deletion \citep{cq-17-focs}.

We define an instance of \refroutine{lazy-inc} at each node in the
balanced binary tree $B$. Whenever a leaf is marked as occupied, the
corresponding edge is inserted into each of $\bigO{\log n}$ instances
of \refroutine{lazy-inc} at the ancestors of the leaf; when a leaf is
marked as unoccupied, it is removed from each of these instances as
well. Each instance of \refroutine{lazy-inc} can then simulate a
weight update over the marked leaves at its nodes in $\bigO{1}$
constant time per instance, plus a total $\bigO{\log n}$ amortized
time. More precisely, the additional time is amortized against the sum
of increases in the logarithms of the weights, which (as discussed
earlier) is bounded above by $\bigO{m \log{n} / \eps^2}$.

We also track, for each canonical forest, the minimum capacity of any
edge in the forest. The minimum capacity ultimately controls the rate
at which all the other edges increase, per
\refequation{bottleneck-weight-update}.

Given a forest $F$ induced by the $\ell$ minimum weight edges of $T$,
we decompose $F$ into the disjoint union of $\bigO{\log n}$ canonical
subforests of $T$. For each subforest we have precomputed the minimum
capacity, and an instance of \refroutine{lazy-inc} that simulates
weight updates on all edges in the subforest. The minimum capacity
over edges in $F$ determines the rate of increase, and the increase is
made to each instance of \refroutine{lazy-inc} in $\bigO{1}$ time per
instance plus $\bigO{\log n}$ amortized time over all instances.

\begin{lemma}
  \labellemma{polylog-weight-update} Given a forest $F$ generated by
  \reflemma{polylog-oracle}, one can update the edge weights per
  \refequation{bottleneck-weight-update} in $\bigO{\log^2 n}$
  amortized time per iteration.
\end{lemma}

Note that \reflemma{polylog-weight-update} holds only for the forests
output by \reflemma{polylog-oracle}. We can not decompose other
forests in $G$, or even other subforests of $T$, into the disjoint
union $\bigO{\log n}$ subforests. \reflemma{canonical-subforests}
holds specifically for the forests induced by the $\ell$ minimum
weight edges of $T$, for varying values of $\ell$. This limitation
highlights the importance of coupling the oracle and the weight
update: the running time in \reflemma{polylog-oracle} for solving
\refequation{relaxation} is amortized against the growth of the
weights, and the weight updates in \reflemma{polylog-weight-update}
leverage the specific structure by which solutions to
\refequation{relaxation} are generated.

We note that the \refroutine{lazy-inc} data structures can be replaced
by random sampling in the randomized MWU framework
\citep{cq-18-soda}. Here one still requires the decomposition into
canonical subforests; an efficient threshold-based sampling is then
conducted at each subforest.

\section{Putting things together}
\labelsection{together} In this section, we summarize the main points
of the algorithm and account for the running time claimed in
\reftheorem{apx-lp}.

\begin{proof}[Proof of \reftheorem{apx-lp}]
  By standard analysis (e.g., \citep[Theorem 2.1]{cq-17-soda}), the
  MWU framework returns a $\apxless$-multiplicative approximation to
  the packing LP \refequation{packing-forests} as long as we can
  approximate the relaxation \refequation{relaxation} to within a
  $\apxless$-multiplicative factor. This slack allows us to maintain
  the weight $\weight{e}$ to within a $\apxpm$-multiplicative factor
  of the ``true weights'' given (up to a leading constant) by
  \refequation{weight-load}.  In particular, we only propagate a
  change to $\weight{e}$ when it has increased by a
  $\epsmore$-multiplicative factor. Each weight $\weight{e}$ is
  monotonically increasing and its growth is bounded by a
  $m^{\bigO{\reps}}$-multiplicative factor, so each weight
  $\weight{e}$ increases by a $\epsmore$-multiplicative factor
  $\bigO{\frac{\log m}{\eps^2}}$ times.

  By \reflemma{polylog-oracle}, each instance of
  \refequation{relaxation} can be solved in $\bigO{\log^2 n}$
  amortized time. Here the running time is amortized against the
  number of weight updates, as the solution can be updated dynamically
  in $\bigO{\log^2 n}$ amortized time. By
  \reflemma{polylog-weight-update}, the weight update w/r/t a solution
  generated by \reflemma{polylog-oracle} can be implemented in
  $\bigO{\log^2 n}$ amortized time. Here again the running time is
  amortized against the growth of the edge weights. Since there are
  $\bigO{\frac{m \log n}{\eps^2}}$ total edge updates, this gives a
  total running time of $\bigO{\frac{m \log^3 n}{\eps^2}}$.
\end{proof}

\section{Rounding fractional forest packings to $k$-cuts}

\labelsection{rounding}

In this section, we show how to round a fractional solution $x$ to
\refequation{lp} a $k$-cut of cost at most twice the cost of $x$. The
rounding scheme is due to \citet{cgn-06} for the more general
problem of Steiner $k$-cuts. The rounding scheme extends the
primal-dual framework of \citet{gw-95,gw-96}. In hindsight, we
realized the primal-dual framework is only required for the analysis,
and that the algorithm itself is very simple.

We first give a conceptual description of the algorithm, called
\algo{greedy-cuts}. The conceptual description suffices for the sake
of analyzing the approximation guarantee. Later, we give
implementation details and demonstrate that it can be executed in
$\bigO{m \log n}$ time.

To describe the algorithm, we first introduce the following
definitions.

\begin{definition}
  Let $F$ be a minimum weight spanning forest in a weighted,
  undirected graph, and order the edges of $F$ in increasing order of
  weight (breaking ties arbitrarily). A \emph{greedy component} of $F$
  is a connected component induced by a prefix of $F$. A \emph{greedy
    cut} is a cut induced by a greedy component of $F$.
\end{definition}

\begin{figure}
  \begin{framed}
    \raggedright \ttfamily
    \underline{greedy-cuts($G = (V,E)$,$\cost$,$x$)} %
    \hfill {\normalfont\emph{(conceptual sketch)}}
    \begin{enumerate}
    \item let $E' = \setof{e \in E \where x_e \geq \frac{n-1}{2 n}}$,
      $E \gets E \setminus E'$
    \item if $E'$ is a $k$-cut then return $E'$
    \item let $F$ be a minimum weight spanning forest w/r/t $x$ with
      $\ell$ components
    \item return the union of $E'$ and the $k - \ell$ minimum weight
      greedy cuts of $F$
    \end{enumerate}
    \caption{A conceptual sketch of a deterministic rounding algorithm
      for $k$-cut. \labelfigure{greedy-cuts-sketch}}
  \end{framed}
\end{figure}

The rounding algorithm is conceptually very simple and a pseudocode
sketch is given in \reffigure{greedy-cuts-sketch}.  We first take all
the edges with $x_e > 1/2 + o(1)$. If this is already a $k$-cut, then
it is a 2-approximation because the corresponding indicator vector is
$\leq \parof{2 - o(1)} x$. Otherwise, we compute the minimum spanning
forest $F$ in the remaining graph, where the weight of an edge is
given by $x$. Letting $\ell$ be number of components of $F$, we
compute the $k - \ell$ minimum weight greedy cuts w/r/t $F$.  We
output the union of $E'$ and the $k-\ell$ greedy cuts.

\citet*{cgn-06} implicitly showed that this algorithm has an
approximation factor of $2 \parof{1 - 1/n}$. Their analysis is for the
more general Steiner $k$-cut problem, where we are given a set of
terminal vertices $T$, and want to find the minimum weight set of
edges whose removal divides the graph into at least $k$ components
each containing a terminal vertex $t \in T$. The algorithm and
analysis is based on the primal-dual framework of
\citet{gw-95,gw-96}. For the minimum weight Steiner tree
problem, the primal-dual framework returns a Steiner tree and a
feasible fractional cut packing in the dual LP. The cost of the
Steiner cut packing is within a $2(1- o(1))$-multiplicative factor of
the corresponding Steiner tree. Via LP duality, the Steiner tree and
the cut packing mutually certify an approximation ratio of
$2(1-o(1))$. The cut packing certificate has other nice properties,
and \citet*{cgn-06} show that the $k-1$ minimum cuts in the
support of the fractional cut packing give a 2-approximate Steiner
$k$-cut.

For the (non-Steiner) $k$-cut problem, we want minimum cuts in the
support of the fractional cut packing returned by the primal-dual
framework applied to minimum spanning forests. To shorten the
algorithm, we observe that (a) the primal-dual framework returns the
minimum spanning forest, and (b) the cuts supported by the
corresponding dual certificate are precisely the greedy cuts of the
minimum spanning forest. Thus \algo{greedy-cuts} essentially refactors
the algorithm analyzed by \citet{cgn-06}.

\begin{lemma}[{\citealp{cgn-06}}]
  \labellemma{greedy-cuts-apx} \algo{greedy-cuts} returns a $k$-cut of
  weight at most $2 \parof{1 - \frac{1}{n}} \rip{c}{x}$.
\end{lemma}
The connection to \citet{cgn-06} is not explicitly clear
because \citeauthor{cgn-06} rounded a slightly more complicated
LP. The complication arises from the difficulty of solving
\refequation{lp} directly for Steiner $k$-cut (which can be simplified
by knapsack covering constraints, in hindsight). Morally, however,
their proof extends to our setting here. For the sake of completeness,
a proof of \reflemma{greedy-cuts-apx} is included in
\refappendix{greedy-cuts-apx}.

\begin{figure}[t!]
  \begin{framed}
    \raggedright \ttfamily
    \underline{greedy-cuts($G=(V,E)$,$\cost$,$x$)}
    \begin{enumerate}
    \item let $E' = \setof{e \in E \where x_e \geq \frac{n-1}{2 n}}$,
      $E \gets E \setminus E'$
    \item if $E'$ is a $k$-cut then return $E'$
    \item let $F$ be a minimum weight spanning forest w/r/t $x$ with
      $\ell$ components
      \begin{blockcomment}
        Arrange the greedily induced components as subtrees of a
        dynamic forest
      \end{blockcomment}
    \item for each $v \in V$
      \begin{enumerate}
      \item make a singleton tree labeled by $v$
      \end{enumerate}
    \item for each edge $f = \setof{u,v} \in F$ in increasing order of
      $x_f$
      \begin{enumerate}
      \item let $T_u$ and $T_v$ be the rooted trees containing $u$ and
        $v$, respectively.
      \item make $T_u$ and $T_v$ children of a new vertex labeled by
        $f$
      \end{enumerate}
      \begin{blockcomment}
        Compute the weight of each cut induced by a greedy component.
      \end{blockcomment}
    \item let each node in the dynamic forest have value $0$
    \item for each edge $e = \setof{u,v} \in E$
      \begin{enumerate}
      \item add $x_e$ to the value of every node on the $u$-to-root
        and $v$-to-root paths
      \item let $w$ be the least common ancestor $u$ and $v$
      \item subtract $2 x_e$ from the value of every node on the $w$
        to root paths
      \end{enumerate}
    \item let $v_1,v_2,\dots,v_{k - \ell}$ be the $k - \ell$ minimum
      value nodes in the dynamic forest. For $i \in [k-\ell]$, let
      $C_i$ be the components induced by the leaves in the subtree
      rooted by $v_i$.
    \item return
      $E' \cup \partial(C_1) \cup \cdots \cup \partial(C_{k-\ell})$
    \end{enumerate}
    \caption{A detailed implementation of a deterministic rounding
      algorithm for $k$-cut. \labelfigure{greedy-cuts}}
  \end{framed}
\end{figure}

It remains to implement \algo{greedy-cuts} in $\bigO{m \log n}$
time. With the help of dynamic trees \citep{st-83}, this can be done
in a straightforward fashion. We briefly describe the full
implementation; pseudocode is given in \reffigure{greedy-cuts}. Recall
from the conceptual sketch above that \algo{greedy-cuts} requires up
to $k-1$ minimum greedy cuts of a minimum spanning forest w/r/t
$x$. To compute the value of these cuts, \algo{greedy-cuts} first
simulates the greedy algorithm by processing the edges in the spanning
forest in increasing order of $x$. The greedy algorithm repeatedly
adds an edge that bridges two greedy components. We assemble a
auxiliary forest of dynamic trees where each leaf is a vertex, and
each subtree corresponds to a greedy component induced by the vertices
at the leaves of the subtree.

After building this dynamic forest, we compute the number of edges in
each cut. We associate each node in the dynamic forest with the greedy
component induced by its leaves, and given each node an initial value
of 0. We process edges one at a time and add its weight to the value
of every node corresponding to a greedy component cutting that edge.
Now, an edge in the original graph is cut by a greedy component iff
the corresponding subtree in the dynamic forest does not contain both
its end points as leaves. We compute the least common ancestor of the
endpoints in the dynamic forest in $\bigO{\log n}$ time \citep{ht-84},
and add the weight of edge to every node between the leaves and the
common ancestor, excluding the common ancestor. Adding the weight to
every node on a node-to-root path takes $\bigO{\log n}$ time
\citep{st-83} in dynamic trees. After processing every edge, we simply
read off the value of each greedy cut as the value of the
corresponding node in the forest. Thus we have the following.
\begin{lemma}
  \labellemma{greedy-cuts-time} \algo{greedy-cuts} can be implemented
  in $\bigO{m \log n}$ time.
\end{lemma}
Together, \reflemma{greedy-cuts-apx} and \reflemma{greedy-cuts-time}
imply \reftheorem{rounding}.

\section*{Acknowledgements}
The author thanks Chandra Chekuri for introducing him to the problem
and providing helpful feedback, including pointers to the literature
for rounding the LP. We thank Chao Xu for pointers to references.

\bibliographystyle{plainnat}
\bibliography{apx-k-cut-compact} %

\appendix

\section{Proofs for \refsection{rounding}}

\labelappendix{greedy-cuts-apx}

\providecommand{\cuts}{\mathcal{C}}

\citet{cgn-06} gave a rounding scheme for the more general
problem of Steiner $k$-cut and the analysis extends to the rounding
schemes presented here. We provide a brief sketch for the sake of
completeness as there are some slight technical gaps. The proof is
simpler and more direct in our setting because we have a direct
fractional solution to \refequation{lp}, while \citet{cgn-06}
dealt with a solution to a slightly more complicated LP. We note that
our analysis also extends to Steiner cuts.  We take as a starting
point the existence of a dual certificate from the primal-dual
framework.

\begin{lemma}[{\citealp{gw-95,gw-96}}]
  \labellemma{mst-duals} Let $F$ be a minimum spanning forest in a
  undirected graph $G = (V,E)$ weighted by $x \in \nnreals^{E'}$. Let
  $\cuts$ be the family of greedy cuts induced by $F$. Then there
  exists $y \in \nnreals^{\cuts}$ satisfying the following properties.
  \begin{mathproperties}
  \item For each edge $e \in E$,
    \begin{math}
      \sum_{C \in \cuts \where C \ni e} y_C \leq x_e.
    \end{math}
  \item For each edge $e \in F$,
    \begin{math}
      \sum_{C \in \cuts \where C \ni e} y_C = x_e.
    \end{math}
  \item
    \begin{math}
      2 \parof{1 - \frac{1}{n}}\rip{y}{\ones} \geq \sum_{e \in F} x_e.
    \end{math}
  \end{mathproperties}
\end{lemma}

We note that the dual variables $y$ can be computed in $\bigO{n}$ time
(after computing the minimum spanning forest).

\begin{lemma}[{\citealp{cgn-06}}]
  \refroutine{greedy-cuts} returns a $k$-cut with total cost
  \begin{math}
    \leq 2 \parof{1 - \frac{1}{n}} \rip{\cost}{x}.
  \end{math}
\end{lemma}

\begin{proof}[Proof sketch]
  Let $y$ be as in \reflemma{mst-duals}. We first make two
  observations about $y$. First, since
  \begin{math}
    x_e \leq \frac{n}{2 \parof{n-1}}
  \end{math}
  for all $e \in E'$, we have (by property (a) of
  \reflemma{mst-duals}) that
  \begin{math}
    y_C \leq \frac{n}{2(n-1)}
  \end{math}
  for all greedy cuts $C$. Second, by \tagr property (c) of
  \reflemma{mst-duals} and \tagr the feasibility of $x$ w/r/t the
  $k$-cut LP \refequation{lp}, we have
  \begin{align*}
    2 \parof{1 - \frac{1}{n}}\rip{y}{\ones}              %
    \tago{\geq}                    %
    \sum_{e \in T} x_e          %
    \tago{\geq}                 %
    k - 1.
  \end{align*}
  Let $C_1,\dots,C_{k-1}$ be the $k-1$ minimum greedy cuts.  Now, by
  \tagr rewriting the sum of the $k-1$ minimum greedy cuts as the
  solution of a minimization problem, \tagr observing that
  $2 \parof{1 - \frac{1}{n}} y$ is a feasible solution to the
  minimization problem, \tagr interchanging sums, and \tagr property
  (a) of \reflemma{mst-duals}, we have
  \begin{align*}
    \sum_{i=1}^{k-1} \sumcost{C_i} %
    &\tago{=}                                               %
      \min{                                           %
      \rip{y'}{\sumcost}                              %
      \where                                          %
      \zeroes \leq y' \leq \ones
      \text{, }                       %
      \support{y'} \subseteq \support{y},               %
      \text{ and }
      \norm{y'}_1 \geq k - 1
      } \\
    &\tago{\leq}                       %
      2 \parof{1 - \frac{1}{n}}\rip{y}{\sumcost} %
      =                         %
      \sum_C y_C \sum_{e \in C} \cost{e} %
      \tago{=}
      \sum_{e \in E'} \cost{e} \sum_{C \ni e} y_C %
      \tago{\leq}                                 %
      \sum_{e \in E'} c_e x_e,
  \end{align*}
  as desired.
\end{proof}

\end{document}